\documentclass[reprint, amsmath,amssymb,aps,prl]{revtex4-2}

\usepackage{amsmath,amssymb,amsfonts}
\usepackage{amsthm}
\usepackage{physics}
\usepackage{graphicx}
\usepackage{dcolumn}
\usepackage{bm}
\usepackage{xcolor}
\usepackage{booktabs, makecell}

\begin{document}

\title{Non-unital noise in a superconducting quantum computer as a computational resource for reservoir computing}

\author{Francesco Monzani}
\author{Emanuele Ricci}
\author{Luca Nigro}
\author{Enrico Prati}
\email{enrico.prati@unimi.it}
\affiliation{Dipartimento di Fisica "Aldo Pontremoli", Università degli Studi di Milano, Via Celoria 16, 20133, Milan, Italy}
\newtheorem{theorem}{Theorem}
\newtheorem{definition}{Definition}
\newtheorem{lemma}{Lemma}

\begin{abstract}
Despite theoretical promises, most machine learning algorithms fail on current quantum computers due to overwhelming noise. Indeed, noise may limit their ability to retain and manipulate information over time, which is especially problematic when dealing with sequential data.
Nevertheless, many recurrent architectures are most efficient if they can process information by progressively reducing the correlation on inputs from earlier time steps. Among them, reservoir computing represents a major paradigm of artificial intelligence for processing time-dependent tasks thanks to the involvement of fading memory. Despite attempts to demonstrate quantum reservoir computing based on intentional perturbations of the network, why the intrinsic noise of the quantum circuit can be straightforwardly exploited remained unknown. We show the feasibility of reservoir computing on a circuit of superconducting qubits based on the dissipation typical of NISQ devices. We prove that noise modeled by a non-unital quantum channel ensures the functioning of a quantum echo state network. In particular, amplitude damping is responsible for drastically improving the short-term memory capacity and reconstruction capability of the network by simultaneously providing fading memory and richer dynamics. Indeed, the damping is capable of preventing the loss of information due to repeated measurements of the circuit. Our experimental results pave the way for the application of reservoir computing methods in non-fault-tolerant quantum computers.
\end{abstract}
\maketitle
\section*{Introduction}
The intrinsic dissipation and decoherence in noisy intermediate-scale quantum (NISQ) computers pose a major limitation in many machine-learning protocols~\cite{PRXQuantum.1.020101}, restricting natural applications of quantum computing in artificial intelligence~\cite{prati2017quantum, Biamonte2017}. Nevertheless, dissipation may serve as a resource for quantum computation~\cite{Verstraete2009,rocutto2021complete,noe2024quantum}. 
Our results show that specific quantum noise of the hardware, namely noise modeled by non-unital quantum channels, significantly improves the short-term memory capacity and expressivity of a quantum network, proving the reliability of quantum reservoir computing as a valid computational architecture for memory-dependent tasks employing a gate-based quantum computer. We emphasize that an in-depth study of the beneficial effects of non-unital noise in quantum machine learning is only in its early stages~\cite{non_uni_prx}, unveiling, for example, a promising application in avoiding barren plateaus in variational problems~\cite{non-uni.barren}. 
Reservoir computing is a well-established supervised machine learning algorithm that employs a fixed neural network -- the reservoir, to process time-dependent information. In addition to initial digital and neuromorphic proposals, respectively, the echo state network~\cite{jaeger2001short} and the liquid state machine~\cite{maass.nat}, reservoir computing has proven effective in a wide range of unconventional deployments~\cite{rc.naka}. Among several physical implementations~\cite{tanaka2019recent},  the exponentially large computational capacity of quantum systems has been harnessed for reservoir computing~\cite{grollier, qrc_opp}, employing photonic circuits~\cite{Spagnolo2022, ph1, ph2}, bosonic oscillators~\cite{nokkala2021gaussian, boso1}, fermionic systems~\cite{ghosh2021realising, ghosh2019quantum}, neutral atoms~\cite{neutral.atoms} and spin networks~\cite{harnessing, naka.2019, chen1, Mujal2023, zamb_dissi,de2023silicon}. Focusing on superconducting quantum computers, we refer to Ref.~\cite{chen2, molteni} for early implementations, more recently empowered by mid-circuit measurements~\cite{kubota2023temporal}. Reservoir computing appears particularly suitable for leveraging the high-dimensional computational space associated with a quantum system \cite{hilb}, as it does not require any training in the parameters of the quantum evolution, thus avoiding typical issues associated with the training of a variational quantum circuit \cite{barren,larocca2025barren}. 
An essential feature for reservoir computing is the so-called fading memory~\cite{maass, monz}. Namely, a reservoir network is most efficient when dealing with the approximation of functional with past dependence if the information is processed by progressively reducing the emphasis on inputs from earlier time steps. Dissipation has recently been recognized as a powerful computational resource in reservoir architectures that employ spin dynamics~\cite{zamb_dissi}. In contrast, phase errors seem to have no utility for quantum computing, since they are theoretically associated with the loss of quantum information. Concerning reservoir computing employing superconducting quantum computers, artificial manual reset of the information flow has proven to be a reliable resource for fading memory in memory-dependent tasks~\cite{molteni,hu2024overcoming, art.res}. In this work, we advance previous approaches by demonstrating why the inherent noise in superconducting quantum processors can be exploited to drive the dynamics required for an effective quantum echo state network. In particular, we show that the damping is essential to guarantee the effective functioning of a gate-based echo state network that operates relying on mid-circuit measurements. Indeed, dissipation can prevent the loss of information due to repeated measurements of the quantum circuit, which steers the internal state of the reservoir toward the maximally mixed state, making the encoding of new input values impossible. Thus, we provide strong theoretical support to our results and to previous empirical findings, which already highlighted the benefit of using amplitude damping in quantum reservoir computing \cite{kubota2023temporal, Domingo2023}. We implement reservoir computing on a circuit of 7 superconducting qubits, and we perform its emulation by including a realistic noise model. By testing our quantum echo state network on standard memory-dependent benchmarks, we recognize which type of noise benefits the performance. Thus, we show that noise modeled by a non-unital channel is needed to simultaneously guarantee memory capacity and accuracy in nonlinear time-dependent tasks. This confirms early results for memory-independent classification tasks~\cite{Domingo2023}. We recall that a quantum channel is unital if it preserves the identity in the operator space. Remarkably, the quantum channel modeling the noise of a superconducting quantum computer falls within such a class of non-unital noise. As typically observed in recurrent networks~\cite{crit_1, pt_zamb}, we identify a critical regime tuned by noise intensity, in which several network capabilities, such as short-term memory capacity and expressivity, are maximized. In this respect, we suitably slowed the execution of the circuit accordingly to maximize the learning accuracy. Finally, as a further theoretical result,  we prove the universality of our gate-based echo state network under the effect of non-unital quantum channels \cite{monz_uni}, enlarging the class of known universal quantum reservoir computers \cite{ nokkala2021gaussian, zamb_dissi, chen2}. We recall that the universality of a neural network architecture is responsible for its computational effectiveness, since it ensures that any input-output mapping can be approximated with arbitrary precision~\cite{hornik.univ.fnn, maass.nat,grig2018}. Our experiment paves the way for real-world application of quantum reservoir computing on noisy intermediate-scale and early fault-tolerant quantum computers.
\section*{Results}
\subsection*{Problem settings}
The task of a reservoir computer for the processing of temporal data consists of defining a mapping -- usually called \textit{filter}, in the framework of reservoir computing, sending an input time series $ u  \in \mathcal{I} $ into another target time series $y \in \mathcal O$, through read-out operations of the information encoded in the reservoir. The schematic structure is represented in Fig.~\ref{tab.1}. Denoting $\mathcal{S}(\mathcal{H})$ the set of density operators $\rho_t$ on the Hilbert space $\mathcal{H}$ that describes some quantum system, a quantum reservoir computer is described by the equations
\begin{align}
    \begin{cases}
        \rho_{t+1} = \mathcal T(\rho_t, u_{t+1})\\
        y_{t+1} = h(\rho_{t+1})
    \end{cases}\,.
\end{align}
expressing the time evolution of the density matrix over discretized time intervals denoted as $\dots\rho_{t-1},\rho_{t},\rho_{t+1}\dots$ Here, $\mathcal T:\mathcal{S}(\mathcal{H}) \rightarrow \mathcal{S}(\mathcal{H})$ is any quantum channel that describes the dynamics of the quantum system and $h$ is a suitable readout operation. The reservoir computer defines the filter, consisting of a unique and causal operator, $\mathcal{C}\colon \mathcal{I}\rightarrow \mathcal{O}$ defined by the mapping
\begin{equation}
y_t = \mathcal C(u)_t = \mathcal{C}\left(u_{|_t}\right) = h\left(\mathcal{T}(\rho_{t-1},u_t)\right)
\end{equation}
where $u_{|_t} = (u_0,\dots,u_{t-1},u_t)$ indicate the input sequence truncated at time $t$. In this work, we employ a gate-based quantum echo state network \cite{jaeger.first.long} to process time-dependent information. We rely on the logical embodiment of the node of the reservoir by the $4^N-1$, where N is the number of qubits, basis elements expressed by the Pauli operators $(I,\sigma_x,\sigma_y,\sigma_z)$ as from the implementation introduced by Ref.~\cite{harnessing}. In this framework, the output nodes correspond to the readout of single qubits $..I\otimes \sigma_{Z(i)} \otimes I,...$.  Its architecture is described in detail in the Methods. 
We test our computational machine by reproducing nonlinear, time-dependent mappings $S(u) = \hat y$ between two time-dependent real sequences $u = \left\{u_t\right\}_{t=0,\dots, L}$ and $\hat y = \left\{\hat{y}_t\right\}_{t=0,\dots, L}$. Specifically, we train the reservoir computer to minimize the distance
\begin{equation}
\mathrm{dist}(\mathcal{C}(u), S(u)) \,.
\end{equation}
To demonstrate the feasibility of quantum reservoir computing on a noisy gate-based quantum computer, we proceed by first proving that some properties of the quantum channel $\mathcal{T}$, specifically non-unitality and contracivtit, satisfied by amplitude damping, ensure the effective functioning of a gate-based echo state network relying on mid-circuit measurements, thus furnishing strong theoretical support to the known results \cite{kubota2023temporal, Domingo2023}. Then, we show the positive impact of amplitude damping with respect to other noisy channels by employing well-established benchmark tasks in the literature of reservoir computing. More specifically, we consider NARMA-n tasks, up to $n=8$, as described in the Methods. Given that realistic noise models indicate non-unital noise as a significant feature of superconducting quantum computers, we proceed with a hardware implementation on an IBM quantum computer to experimentally validate the learning capabilities.
\begin{figure*}[htb]
\centering
\includegraphics[width=\textwidth]{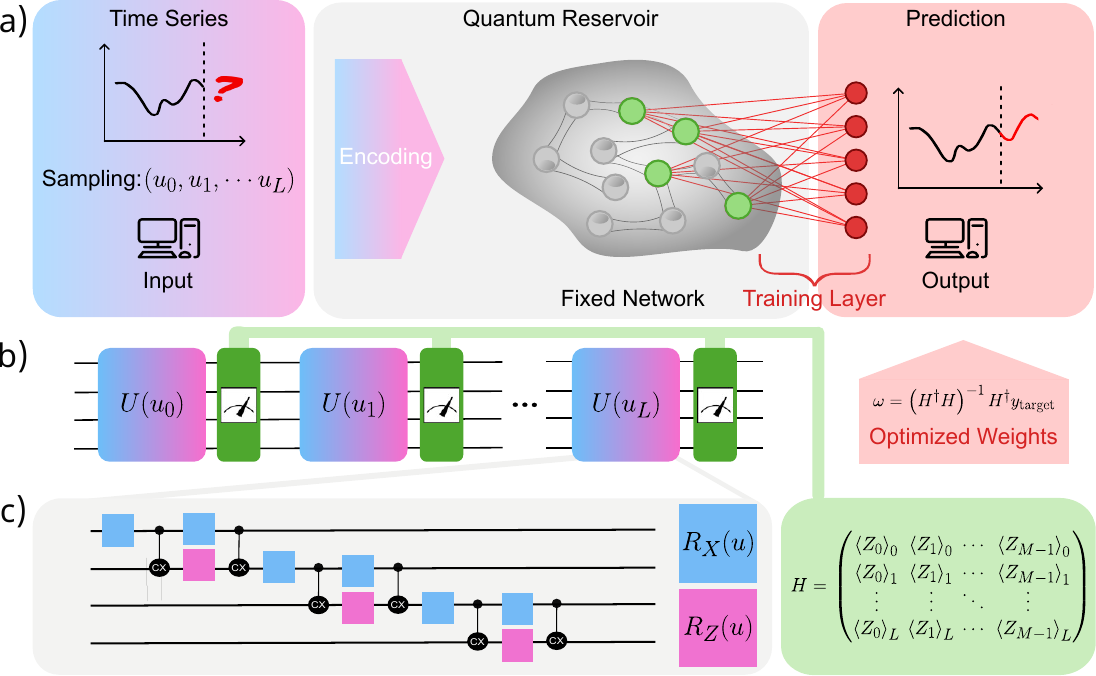}
    \caption{\textbf{The architecture of quantum reservoir computing.} a) The time-dependent information in the input is sampled to produce the input time series. The information is stored in the logical nodes described by the Pauli basis associated with the qubits. During the evolution of the reservoir, the information is extracted by mid-circuit measurements of the true nodes (in green) of the network. The reservoir signal is the outcome of the measurements. The only trained part of the system is the readout weights $\omega$ (red connections). b) The circuit is employed as the reservoir of the quantum echo state network. The input values are encoded in parametric unitary gates. Mid-circuit measurements are performed to extract the reservoir signal online. After the execution, the optimal weights of the linear readout are computed through the pseudoinverse of the overall reservoir signal along the training interval, in order to minimize the distance from the target time series. c) The precise architecture of the gates in each unitary layer. The inputs are encoded in the composition of rotations along the $X$ (in blue) and $Z$ (in pink) axis of the Bloch sphere. Z Pauli operators are measured as true nodes. CNOT operators allow for the entanglement of the qubits. The qubits are left in the collapsed state after each measurement. In the work, we fix $N=7$ as the number of qubits. Finally, the reservoir signal, namely the expected value of each Z operator, is collected in the matrix $H$.}
    \label{tab.1}
\end{figure*} 
\begin{figure*}[htb]
\centering
\includegraphics[width=0.95\textwidth]{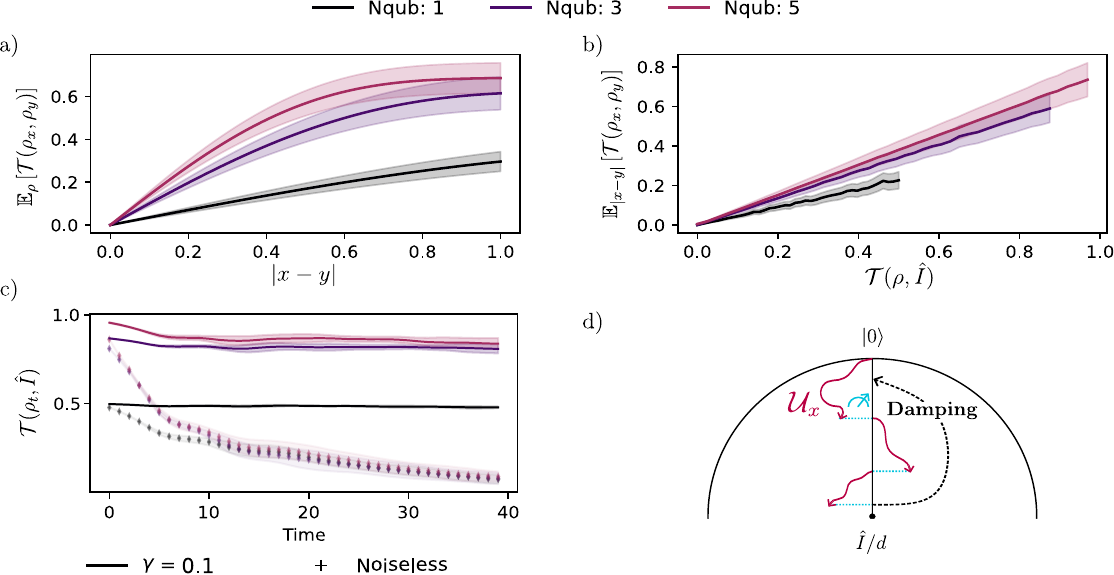}
    \caption{\textbf{Information retention induced by damping.} a-b) Mean trace distance of the internal state of the reservoir after input injection for, respectively, (a.) increasing distance between inputs and (b.) increasing trace distance from $\hat{\mathcal{I}}$, recalling that this is bounded by $1-\frac{1}{2^{N_a}}$. In the first case, the average is computed over 100 initial states $\rho_0$, with distance from $\hat{\mathcal{I}}$ sampled from a uniform distribution $U(0, \frac{1}{2^{N_a}})$.  In the second case, the average is computed over 100 uniformly sampled input values. c) Trace distance of the circuit from the $\hat{\mathcal{I}}$ along 40 time steps of input encoding and mid-circuit measurements, averaged over 10 repetitions with random inputs. If the dynamics is affected by amplitude damping (full line), the state remains deviated from $\hat{\mathcal{I}}$ during the evolution d)  A sketch of the evolution of the internal state of a 1-qubit circuit reservoir during three steps of input encoding and midcircuit measurements. The unitary evolution does not modify the distance from the origin of the Bloch sphere, while the measurement collapses the state, reducing the trace distance from the maximally mixed state. The damping prevents the information loss by steering the state towards $\ket{0}$.}
    \label{tab.sep}
\end{figure*} 
\subsection*{Quantum reservoir computing with mid-circuit measurements}
An approach to gate-based quantum reservoir computing, also pursued in this work, involves dealing with a fully accessible reservoir of $N_a$ qubits \cite{rep, kubota2023temporal}, without any reset of the qubits, and thus exploiting quantum trajectories as a memory retention mechanism. Specifically, the input encoding is performed through parametric unitary matrices, and the information is decoded by measuring the qubits online along the execution of the circuit by exploiting mid-circuit measurements. After each measurement, the qubits are let evolve starting from the collapsed state. Thus, information retention relies on the statistics of the observables of the accessible qubits only, without the need for a hidden reservoir, as described in the Methods. Nonetheless, it is well known that repeated measurements on a quantum system progressively induce information loss. This hinders the computational capabilities of the reservoir computer, as we discuss later in this Section. Our work shows how the controlled damping of the qubits can overcome this issue, making the algorithm stable over time intervals significantly longer than the typical information survival time in the measured system. 
In particular, this work proves that the effect of damping compensates the quantum information loss due to mid-circuit measurements, thereby ensuring the persistence of the algorithm. We therefore explain also previous empirical results on the beneficial effect of non-unital noise on the functioning of a gate-based echo state network for temporal processing \cite{kubota2023temporal, rep}. 
\subsection*{Damping prevents the decoherence of reservoir state}
Now we proceed by providing such theoretical explanation of the need of damping for the functioning of a gate-based echo state network. A fundamental prerequisite for the effectiveness of a reservoir computer is its capacity to separate distinct inputs reliably. Thus, the reservoir dynamics must be sufficiently expressive to produce distinguishable states $\rho_x = \mathcal{U}_x \rho$ and $\rho_y = \mathcal{U}_y \rho $ after input injection. Here, $\mathcal{U}_j\,,\, j = x,y$ denotes the unitary encoding of the input value $j$, described in detail in the Methods. The separability ensures that the internal states of the reservoir encode meaningful distinctions in the input, which is essential for downstream tasks such as classification or prediction. As stated by Holevo–Helstrom theorem \cite{Helstrom1969, holevo1973bounds}, the distinguishability of two different quantum states may be quantified by their trace distance, namely $\mathcal{T}(\rho,\sigma) = \frac{1}{2}\norm{\rho-\sigma}_1$ with $\norm{A}_1 = \text{Tr}\left(\sqrt{A\,A^\dag}\right)$. By continuity of the unitary evolution, the difference in trace norm of two internal states of the reservoir, corresponding to distinct values of the input, is related to the distance between the initial state $\rho$ and the maximally mixed state $\hat{\mathcal{I}}/d$, namely
\begin{align}\label{continuity}
    \mathcal{T}\left(\rho_x - \rho_y\right) \le C(U)\cdot C(|x-y|)\cdot \mathcal{T}\left(\rho - \hat{\mathcal{I}}/d \right)
\end{align}
Here, we are denoting with $C(U)$ and $ C(|x-y|)$ two constants that depend on the unitary and the distance input values, respectively, and with $d = 2^N$ the dimension of the Hilbert space. We quantify the expression in Eq. \eqref{continuity} in Fig.~\ref{tab.sep}, showing that a larger deviation from $\hat{\mathcal{I}}$ leads to enhanced separability between distinct inputs. Nonetheless, the repeated mid-circuit measurement leads to quantum decoherence, and the internal state $\rho_t$ is driven toward the maximally mixed state, $\rho_t \rightarrow \hat{\mathcal{I}}/d$.  Indeed, any unitary rotation does not increase the distance from the identity, since
\begin{align}
\norm{\mathcal{U}_x \rho - \hat{\mathcal{I}}/d}_{1} = \norm{\mathcal{U}_x( \rho - \hat{\mathcal{I}}/d)}_{1} = \norm{\rho - \hat{\mathcal{I}}/d}_1
\end{align}
while projective measurements decrease it, namely
\begin{align}
    \norm{\Pi \rho \Pi^\dag - \hat{\mathcal{I}}/d}_1 =  \norm{\Pi (\rho - \hat{\mathcal{I}}/d )\Pi^\dag }_1 \le \norm{\rho - \hat{\mathcal{I}}/d }_1\,.
\end{align}
Here the equality is realized if and only if \mbox{$\rho = \bigotimes^{N_a} |i\rangle \langle i|$}, with $i = 0,1$, namely is the tensor product on the eigenvalues of the projective measure. Consequently, the dynamics produces progressively less distinguishable states. Eventually, after an effective time window, the state of the reservoir becomes independent from the input, since $\mathcal{U}_x\,\hat{\mathcal{I}} = \hat{\mathcal{I}}$ for any input value $x$. 
In the following, we demonstrate how dissipation prevents the loss of information, enabling stable computation over an extended period. In particular, the concentration of the reservoir in the maximally mixed state $\hat{\mathcal{I}}$ can be prevented by exploiting the action of a non-unital channel. We recall that a quantum channel $\mathcal E$ is non-unital when $\mathcal E (\hat{\mathcal{I}}) \!= \hat{\mathcal{I}}$. As an example, the amplitude damping channel can be adopted for this purpose, as it tends to relax the qubit state toward $\ket{0}$, with a decay rate denoted with $\gamma$ throughout the work. Moreover, if each qubit is affected by damping, the effect of dissipation is not hindered by the dimensionality of the circuit, as shown in Fig.~\ref{tab.sep}.
\subsection*{Universal approximation property ensured by non-unital and contractive channel}
As a further theoretical support explaining the crucial role of damping, or other similar non-unital channels, we demonstrate that its effect guarantees the universal approximation property of the gate-based echo state network considered in this work and Ref. \cite{kubota2023temporal}. 
A family of computational architectures is universal for a certain set of functionals if, for any of these mappings, there exists an instance within the family that approximates it with arbitrary precision \cite{hornik.univ.fnn}. 
In particular, the universality of a family of reservoir computers for fading memory mappings relies on some general sufficient conditions \cite{maass.nat, monz_uni}. 
Precisely, a reservoir computer is universal for temporal tasks if the internal dynamics of the reservoir has fading memory and can discriminate between distinct inputs. The formal theorem is detailed in the Supplementary Material attached to this work. 
Then, the polynomial subalgebra of functionals is constructed by spatial multiplexing, as shown in Ref. \cite{chen2}. In particular, we recall that a noise modelled by a non-unital and strictly contractive channel ensures the fading memory of the reservoir computer, as already proved in Ref. \cite{kubota2023temporal}. 
Moreover, we prove in the Appendix that even separability is ensured when the quantum channel has a unique fixed point away from the identity. A similar result was stated and proved in Ref.~\cite[Theorem 2]{qrc} . This is the case, for example, for the amplitude damping channel. Thus, we prove that a gate-based quantum echo state network is universal under the effect of a class of noise that includes the amplitude damping, thus providing further support to the empirical findings. 
We refer to the Supplementary Materials for the detailed mathematical proof. 
\subsection*{Impact of noise on the fading memory}
The action of noise provides fading memory to the system by gradually reducing the network's dependence on past inputs. Indeed, the action of the quantum channel sequentially modifies the state of each qubit, contributing to the loss of information over time. We remark that fading memory is essential in reservoir computing, as it enables the network to integrate online new information.
To assess how increasing noise affects the network's dependence on past inputs, we analyze its impact on the correlation between consecutive qubit measurements. See Methods for a detailed description of the architecture of our echo state network. After repeating the execution for $S\sim10^5$ shots, to reconstruct the average value of the $\sigma_{Z(i)}$ Pauli observables, we compute, for any measured qubit,  
\begin{equation}
    \mathrm{corr}(\overline{z}^i_t,  \overline{z}^i_{t-k}),
\end{equation} where the vector $\overline{z}^i_t = \left\{(z_t)_s\right\}_{s=1,\dots,S}$ contains the outcomes of measurements for each repetition of the experiment. Then, we compute the mean correlation over the whole execution by averaging over $t$. Fig.~\ref{tab.2} shows the mean correlation over time for increasing time windows $k$. As expected, higher noise intensity boosts the progressive loss of correlation. The loss of correlation appears comparable for different noise models. Thus, fading memory alone does not a priori justify the beneficial effect of amplitude damping on the short-term memory capacity of the network and its accuracy in learning tasks. From a theoretical point of view, fading memory is related to the contractivity of the evolution of the echo state network \cite{grig2018}. In our model, the contractivity in trace norm is ensured by the action of noise, as discussed in the Supplementary Material.
 \begin{figure*}[htb]
\centering
\includegraphics[width=0.9\textwidth]{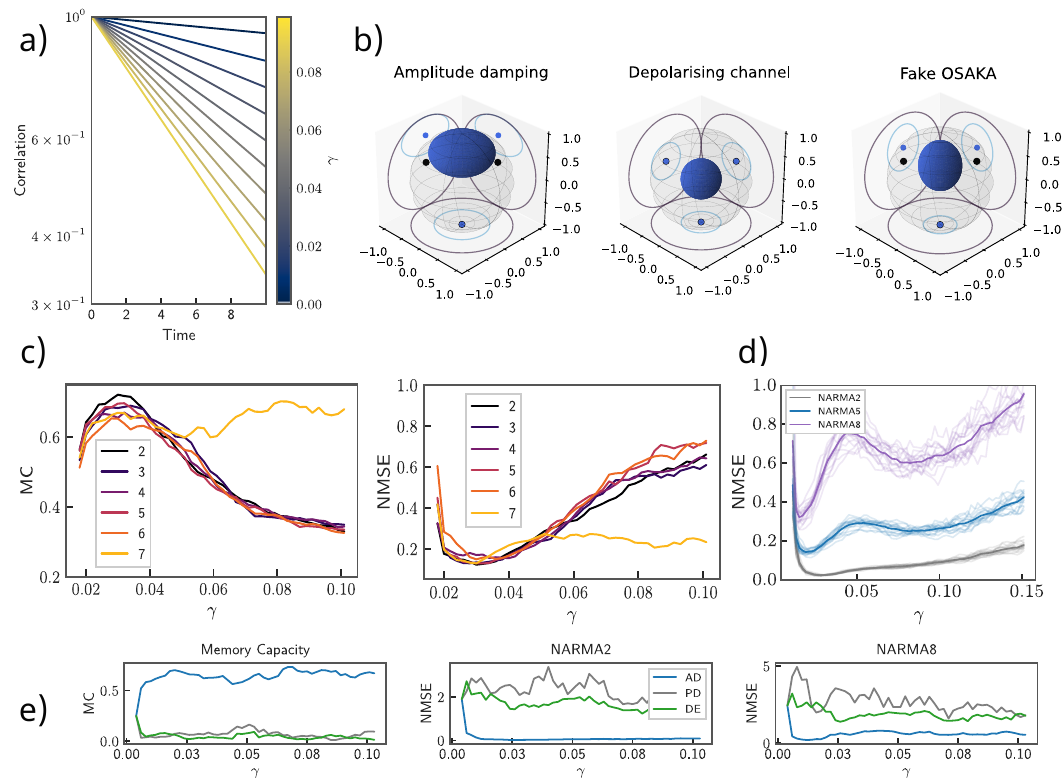}
    \caption{\textbf{Effect of noises on learning.} a) Decay over time in the mean correlation in subsequent measurements is more pronounced when amplitude damping intensity increases. Yellow lines indicate a higher intensity of the noise. The thickness of lines is proportional to the accuracy in NARMA5 prediction. b) From left to right, the blue area is the image of the quantum channel on the Bloch sphere. Black and blue dots are the projections of the sphere's center, respectively, before and after the application of the channel. For clarity, simulations are for $\gamma = 0.6$ and tenfold intensity of noise in the IBM\_OSAKA backend. c) Accuracy in NARMA-8 task and short-term memory capacity for various numbers of measured qubits in the 7 qubits register, under the effect of increasing amplitude damping. $\gamma = 0$ amounts for noiseless execution of the circuit. Yellow lines correspond to complete measurements. Running average on a window of ten values of $\gamma$ is performed for smoother visualization. d) Accuracy in NARMA-n tasks under the effect of increasing amplitude damping.  Here, the reservoir is made of 7 qubits, and all the qubits are measured in the experiment. Results for a swarm of 20 executions are plotted in transparency, mean values are in bold in the picture. Running average on a window of ten values of $\gamma$ is performed for smoother visualization. The optimal noise regime corresponds to $\gamma \sim 0.03$. e) Accuracy in reproducing NARMA tasks and memory capacity for increasing intensity of amplitude damping (in blue), phase damping (in grey), and depolarizing channel (in green). The performance of the noiseless execution corresponds to $\gamma = 0$. Only amplitude damping ensures the learning of the network. Again, the complete measurement of a reservoir with 7 qubits is employed}
    \label{tab.2}
\end{figure*} 
\begin{figure*}[htb]
\centering
\includegraphics[width=0.9\textwidth]{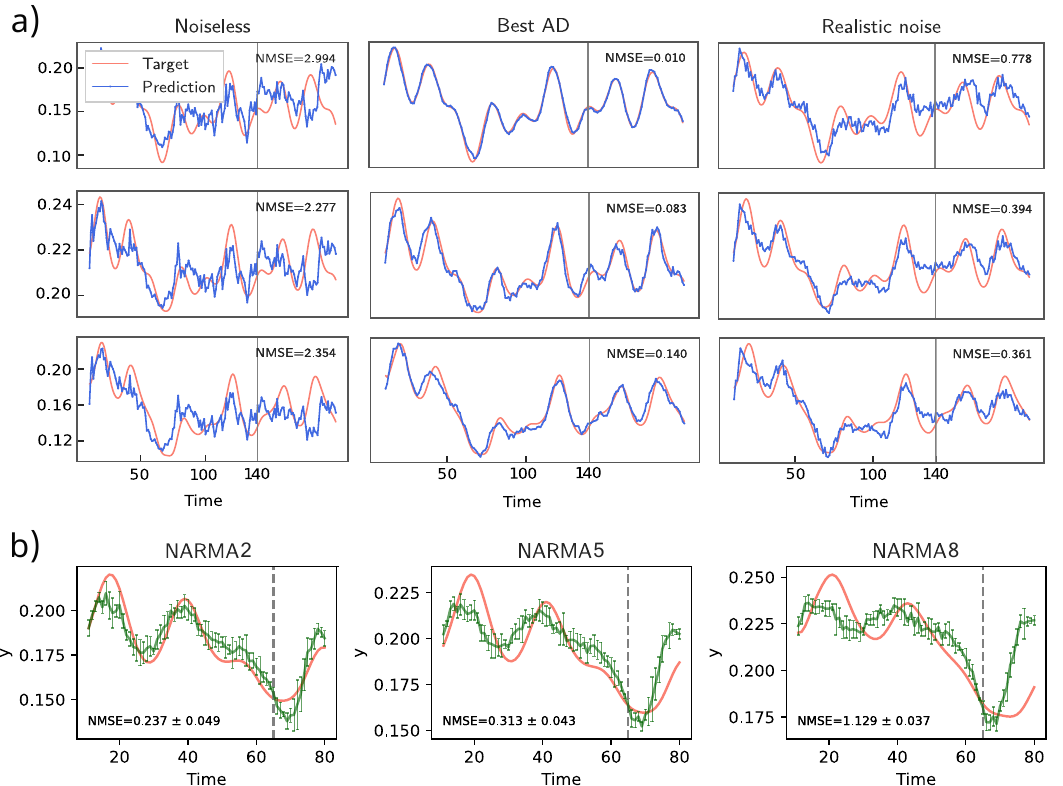}
    \caption{\textbf{Reconstruction of NARMA tasks on simulators and quantum hardware.} a) The reconstruction of NARMAn tasks for, respectively, the noiseless execution of the circuit, the execution affected by optimal amplitude damping, and the realistic noise emulating the one in IBM\_OSAKA quantum backend. Optimal values of noise are $\gamma = 0.024, 0.03, 0.014$ for NARMA2, NARMA5 and NARMA8, respectively. The grey vertical line separates the training and test intervals in the output sequence. The amplitude damping ensures the best learning of the system, which is slightly deteriorated under the composition of several noisy quantum channels representing the realistic noise.  We highlight that, in the noiseless execution, the target is met for around 50 time-steps. After that, the internal state of the reservoir is driven to the maximally mixed state, becoming unresponsive as discussed in Fig. \ref{tab.sep}. b) The experimental results of reconstruction of NARMA2, NARMA5 and NARMA8 (from left to right) tasks employing IBM\_BRISBANE quantum backend. The experiment is repeated 5 times. The mean results are reported in the green line, with error bars indicating a $1\, \mathrm{s.d.}$ interval. The input length is reduced, in order to reduce the depth of the circuit to fit the technical limitations of the quantum hardware. Precisely, we sample 80 values for $u_t$. Reconstruction deteriorates for increasing the depth in memory required, indicating the necessity of the implementation of error mitigation techniques that allow for reducing the unital noise in the evolution.}
    \label{tab.3}
\end{figure*} 
\subsection*{Effect on learning of simulated quantum noise models}
In this Section, we study the effect of different noise models on the memory capacity and the learning capability of a gate-based echo state network \cite{kubota2023temporal}.  Concerning the simulation of diverse models of quantum noise, their effect on the learning capabilities of the network is investigated by systematically tuning its intensity in numerical experiments. Throughout the whole Results section, the parameter that tunes noise intensity is referred to as $\gamma$. Our results confirm previous findings \cite{kubota2023temporal, Domingo2023}, showing that amplitude damping is drastically improving the performance of the reservoir computer.
 After confirming that any contractive noisy channel ensures the fading memory of the network, we investigate their positive on the predictive capabilities of the architecture using two standard benchmarks for recurrent networks, namely the measure of short-term memory capacity (MC) and the non-linear autoregressive moving average (NARMA) task, presented in detail in Subsection "Benchmark for performance analysis". The short-term memory capacity measures the ability of the network to reproduce online a delayed window of the time-dependent input, while the NARMA task requires approximating a non-linear functional that depends on a fixed amount of the recent information in the input time series. Three noise models are applied separately, namely amplitude damping, phase damping, and depolarizing. Among them, only the amplitude damping (AD) is non-unital. Remarkably, we observe that the action of the only amplitude damping drastically enhances the memory capacity of the network. Consistently, the higher memory capacity is reflected in better accuracy in reproducing the time-dependent task. On the other hand, phase damping (PD) and depolarizing (DE) noise have detrimental effects on both memory capacity and accuracy. These results are summarized in Fig.~\ref{tab.2}, showing the normalized mean square error (NMSE) and MC for increasing the intensity of the three noises considered. The beneficial effect of amplitude damping is justified by its non-unital nature as a quantum channel. Indeed, the action of a non-unital channel enriches the dynamics of the reservoir allowing for a wider exploration of the phase space and, consequently, greater storage of information. We discuss more details about non-unital channels and theoretical insight into the role of unitality in quantum reservoir computing in the Supplementary Material. To anticipate the behavior on NISQ hardware, we repeated the experiment by applying a channel that replicates the noise model of the backends provided by IBM. A realistic noise model for current superconducting hardware consists of the application of the composition of quantum channels, including, but not restricted to, the three noise models considered in this work. Thus, recalling that quantum channels are non-expanding by definition \cite{qrc}, we may assume that the dynamics is strictly contractive since the composition of non-expanding channels and strictly contractive channels is in turn a contractive channel. Moreover, we confirmed that the realistic noise is non-unital by analyzing its effect on the Bloch sphere, as shown in Fig.~\ref{tab.2}. As expected, the network still shows a remarkable improvement in learning capabilities when compared with noiseless executions, even if the accuracy in reproducing tasks is slightly deteriorated compared to optimal pure amplitude damping. The results in the reconstruction of the NARMA task are reported in Fig.~\ref{tab.3}. 
\subsection*{Impact of partial and global measurements on the optimal noise regime}
By tuning the intensity of pure amplitude damping in simulations, the learning capabilities of the network are maximized for $\gamma \sim 0.02$. corresponding to which the echo state performs best for any task. This suggests the existence of an optimal critical regime for the network's expressivity. 
The emergent properties of a quantum circuit depend on the rate of mid-circuit measurements performed on it \cite{koh2023}. This turns out to be the case also for the expressivity and the short-term memory capacity of the circuit acting as the reservoir. Indeed, when performing a complete measurement of the qubits in the register, the effect of amplitude damping appears stable under increasing noise intensity, at least for small values of $\lambda$ (approximately, $\gamma< 0.1$). As reported in Fig.~\ref{tab.2}, this optimal value of $\gamma$ is followed by a plateau of suboptimal values, for which the performance remains acceptable. In the case of partial measurements of the qubit register, while the optimal noise regime persists for the same noise intensity with no deterioration in memory capacity or predictive accuracy, the plateau of suboptimal values is immediately disrupted by increasing the noise intensity. Moreover, the best performances of the network essentially do not depend on the number of measured qubits.
\subsection*{Experiment on superconducting quantum computer}
We confirm our findings by experimenting on the IBM-Brisbane superconducting quantum computer, guided by the results of the noisy simulations. First, we employ the complete measurement of the register in order to improve stability under various noise intensities. As a further optimization, we introduce a tunable delay in executing rotations, deploying them as $ RZ(\alpha) = RZ(\alpha/d)\circ \dots \circ RZ(\alpha/d)$ and the same for $RX$, where $\alpha$ denotes the angle of the desired rotation. Then, assuming that a delayed execution enlarges the amount of noise affecting the circuit evolution, we optimize the delay parameter $d$ by employing the emulator, before running the experiment on quantum hardware. However, there is a trade-off to consider. Although using multiple rotations instead of a single rotation on a quantum computer helps us achieve the desired increase in error, it also comes with a downside. Quantum computers have a limited allowance for the number of operations that can be submitted, which constrains our approach. We found an optimal compromise between performance and depth of the circuit for $d=4$. The experiment is conducted 5 times and the average results are shown with error bars of 1 standard deviation in Fig. \ref{tab.3}. The experiment shows that the learning capabilities are considerably better when compared to the ideal execution without noise performed on the simulator. However, the NARMA-8 task clearly shows that when processing longer information windows is required, the unital component of the intrinsic noise in the quantum computer dominates, leading to a progressive deterioration in performance. This highlights the need to implement appropriate error mitigation techniques to eliminate the effects of unital noise or to adopt suitable strategies to amplify the impact of non-unital noise, which is undoubtedly recognized as essential in simulations.
\section*{Discussion}
We demonstrate the beneficial effect of damping in gate-based quantum reservoir computing for time-dependent tasks. 
Firstly, we showed that dissipation is essential for the functioning of a gate-based echo state network. Indeed, the repeated measurements needed for information decoding naturally steer the internal state of the circuit towards the maximally mixed state, making the input encoding ineffective after a short time window. 
In this work, we show that non-unital noise, as amplitude damping, prevents this phenomenon, leading to a long-term effective reservoir computation. Moreover, as a further theoretical support to empirical findings, we identify the properties of the quantum dynamics of the reservoir that ensure the universal approximation property of the reservoir computer. 
Indeed, they turn out to be the same properties ensuring the functioning of the reservoir computer. 
Based on our theoretical finding, we confirm early empirical results \cite{kubota2023temporal, Domingo2023} by observing through numerical simulations that amplitude damping ensures great accuracy in benchmark time-dependent tasks. 
The learning is confirmed by conducting the experiment on the superconducting computer. Non-unitality of the evolution of the quantum circuit has already been recognized as a key resource for reservoir computing \cite{hu2024overcoming}. 
Here we show that it can be naturally induced by employing intrinsic noise in the system. Moreover, the network's response to noise qualitatively varies depending on whether complete or partial measurements are performed. 
In particular, in the case of complete measurements, the learning capabilities of the network remain effective over a sufficiently broad range of noise intensity values, confidently encompassing the amount of noise present in current quantum hardware. These facts pave the way for use-case results in real-world applications. 
Moreover, from the perspective of using quantum hardware, our results offer valuable insights that can inform best practices for developing and applying suitable error mitigation techniques. 
On the one hand, our observations suggest prioritizing the mitigation of dephasing and readout errors for the application of quantum reservoir computing protocols while ignoring the correction of dissipation effects. Partially effective error mitigation techniques may be employed to reduce the number of measurements since the optimal performances of the network persist in the case of partial measurement of the circuit. 
As a consequence, employing error control algorithms that allow positioning within this regime can lead to significant savings in memory consumption and execution time on quantum hardware by reducing the total number of mid-circuit measurements. 
Nevertheless, the relatively poor experimental results are due to the fact that current quantum computers already have error rates below the noise regime that ensures the functioning of an echo state network gate-based with mid-circuit measurements. 
For this reason, new algorithms that increase qubit damping are needed for effective implementation.
\section*{Methods}   
\subsection*{Architecture of the quantum echo state network}
Our gate-based echo state network employs an N-qubit circuit to encode and process time-dependent information. Then, a linear readout is trained to replicate a mapping $S(u) = \hat y$ between an input sequence $u = \left\{u_t\right\}_{t=0,\dots, L}$ and a target sequence $\hat y = \left\{\hat{y}_t\right\}_{t=0,\dots, L}$. 
In each execution, the time series is entirely processed by executing a single circuit that consists of unitary evolution, which is used for encoding the input value, and mid-circuit measurements of the qubits, whose outcomes are used for computing the expected value of the observables employed as readout. 
In particular, after the encoding of an input value, the mid-circuit measurement of the register is performed on the computational basis and, as opposed to other approaches \cite{hu2024overcoming, xiong2025fundamental}, the system is let evolve in the collapsed state after the measurements. 
The entire circuit, which encodes progressively the input time series, is executed for $S$ shots for computing the deterministic expected value of the observables used as readout. 
We describe in the following the one-step evolution of the systems, which is replicated for each value of the input time series. 
We denote with $\rho_t$ the density that describes the state of the N-qubit register at time-step $t$.
We remark that the density matrix $\rho_t$ depends formally on any value of the input $u_{t'}$ with $t' \le t$, when considered as a statistical mixture of the $S$ repetitions of the circuit, that are conditioned on the outcome of the mid-circuit measurements. Indeed, the outcome probabilities directly depend on the value of the input. Firstly, $\rho_t$ is subject to a unitary evolution affected by some noise, modelled as a quantum channel applied after the unitary evolution, namely 
 \begin{equation}\label{uno}
      \widetilde{\rho}_{t+1} = \mathcal{E}\left(U(u_{t+1}) \rho_t U^{\dag}(u_{t+1}) \right)\,.
 \end{equation} 
Here, $\mathcal{E}$ denotes the quantum channel that models the quantum noise, described later in the paper. The input $u_{t+1}$ is encoded in the angles of some rotations which constitute the unitary gate $U$. The circuit, described in detail in Fig~\ref{tab.1}, is composed of rotations RX, RY, and CNOT gates, which entangle pairs of qubits at a time. 
After the input injection, a mid-circuit measurement of the qubits is performed, and the system is let evolve in the collapsed state. 
Thus, for each execution of the circuit, the evolution is conditioned by the measurement outcomes. This mechanism is key to integrating past information into the system, enabling the processing of time-dependent data without requiring hidden qubits for memory retention, in contrast to the approach used in Ref. \cite{hu2024overcoming}. 
In particular, we assume that each qubit is measured separately and simultaneously, so that each measurement does not affect the outcome of the others. Formally, we write the effect of the measurement of $M$ out of $N$ qubits as
 \begin{align}
     \rho_{t+1} = \Pi_M \widetilde{\rho}_{t+1} \Pi_M^\dag 
 \end{align}
 where
 \begin{align}
     \Pi_M = \bigotimes_{s=1}^M |m_{s,t+1}\rangle\langle m_{s,t+1}| \otimes \mathbb{I}_{N-M} 
 \end{align}
denoting $m_{s,t+1}\in \left\{0,1\right\}$ the outcome of the measurement of the qubits $s$ at time $t+1$.
In particular, we measure either all the qubits employed in the circuits ($M=N$) or a smaller subset ($M<N$). In both cases, the reservoir signal at time-step $t$ is constructed online as the expectation value of the $Z_i$ Pauli operators, corresponding to the measured qubits, namely 
 \begin{equation}
     z_t = \left[\langle \sigma_{Z(1)} \rangle_{\rho_t}, \dots, \langle \sigma_{Z(M)}\rangle_{\rho_t}\right]^T \,.
 \end{equation}
 Explitcly, we construct the expectation value of the observable $\sigma_{Z(i)}$ by collecting the $S$ outcomes of each measurement of the $i^{th}$ qubit, and then computing 
 \begin{equation}
      \langle \sigma_{Z(i)} \rangle_{\rho_t} = \frac{1}{S}\left(\sum1_{m_{i,t}= 0} - \sum1_{m_{i,t}=1} \right),
 \end{equation}
 where $1_{m_{i,t}= k}$ assumes value $1$ when $m_{i,t}= k$ and $0$ otherwise, with $k = 0,1$, since indeed the eigenvalues of the observable $\sigma_Z$ are $1$ and $-1$, corresponding to the eigenstate $\ket{0}$ and $\ket{1}$, respectively.
 This procedure amounts to computing the expected value of the global observables $\mathbb{I}^{\otimes i-1}  \otimes \sigma_{Z(i)} \otimes \mathbb{I}^{\otimes N-i}$. This is consistent with the assumption about the effect of local measurements since these observables commute.
\subsection*{Training of linear readout}
Reservoir computing aims to best reproduce a given nonlinear map $S(u) = \hat y$ by training the linear readout $h$. We will give examples of such learning tasks $S(u)$ later in the work. After an initial washout period $T_{\text{wo}}$, required to erase the dependence on the initial state of the network, the reservoir signal is extracted until the end of the training interval $T_{\text{tr}}$ and it is stored in the matrix $H = \left\{z_t\right\}_{t=T_{\text{wo}} + 1,\dots, T_{\text{tr}}}$.  After the quantum evolution, the optimal readout weights are computed. Precisely, the weights of the linear readout $y = H w$ are chosen as the set of parameters $w$ that minimizes the distance $\norm{\hat{y} - y}$. It can be easily computed by linear regression or exploiting the Moore-Penrose pseudoinverse matrix, namely $ w = (H^T H)^{-1} H^T \hat y_{\text{tr}}$, where $\hat y_{\text{tr}} = \left\{\hat y _t\right\}_{t = T_{\text{wo}} + 1,\dots, T_{\text{tr}}}$ is the portion of target sequence used for the training. We refer to Ref.~\cite{molteni} for further details on the training algorithm. After the training, the predicted values of $y$ are given by 
\begin{equation}
    y_t = z_t \cdot w \qquad \text{for any } t \in [T_{tr} + 1, L] \,.
\end{equation}
It is worth mentioning that, thanks to the intrinsic multitasking nature of reservoir computing, the same reservoir signal can be exploited to reproduce different maps, only repeating the training to compute the optimal weights. This fact allows for the reuse of information stored in the reservoir for tasks that share the same input, notably reducing execution time. 
\subsection*{Benchmark for performance analysis}
\begin{table*}[tb]
   \centering
  \begin{tabular}{>{\raggedright\arraybackslash}m{2.4cm} 
                  >{\raggedright\arraybackslash}m{3.1cm} 
                  >{\raggedright\arraybackslash}m{9.5cm}}
  \hline
  \textbf{Noise model} & \textbf{Description} & \textbf{Kraus decomposition}\\
  \hline
  Amplitude damping &  Loss of energy to the environment & $K_0 = \begin{pmatrix}
  1 & 0 \\
  0 & \sqrt{1-\gamma}
\end{pmatrix},\,\, K_1 = \begin{pmatrix}
  0 & \sqrt{\gamma}\\
  0 & 0
\end{pmatrix}$
\\ 
\addlinespace
  Phase damping & Loss of quantum information & 
  $K_0=\sqrt{1-\gamma}\,\mathbb{I}, \,\, K_1 = \begin{pmatrix}
  \sqrt{\gamma} & 0\\
  0 & 0
\end{pmatrix},$  $ K_2 = \begin{pmatrix}
  0 & 0\\
  0 & \sqrt{\gamma}
\end{pmatrix}$
 \\
 \addlinespace
  Depolarizing & Decay to maximally mixed state & $K_0 =\sqrt{1-\gamma}\,\mathbb{I}, \, \, K_1 = \sqrt{\gamma}\,X$ 
  $K_2 = \sqrt{\gamma}\,Y,\,\,K_3 = \sqrt{\gamma}\,Z$\\ 
  \hline
\end{tabular}
\caption{The Kraus decomposition of the quantum channel describing respectively amplitude damping, phase damping, and depolarizing noise.}\label{noises.tab}
\end{table*}
 Our experiments aim to demonstrate the beneficial effect of intrinsic noise in a superconducting quantum computer for reservoir computing. First, to verify whether quantum noise improves the fading memory, we study its effect on the correlation among several measurement outcomes of the same qubit over time. Then, the performances of our computational architecture under the effect of quantum noise are analyzed by exploiting standard benchmarks for recurrent neural networks. We measure the short-term memory capacity of our network using the memory capacity metric, which quantifies the amount of variance in the delayed input that can be recovered from the trained output \cite{jaeger2001short}. To compute it, the system is required to replicate an input sequence $\left\{u_t\right\}_{t=0}^{T}$ delayed by a time interval $d$. Namely, the system is required to approximate the functional $\hat y_t = u_{t-d}$. Then, the $d$-memory capacity is defined as
\begin{equation}
\mathrm{MC}_d =  \frac{\mathrm{Cov}^2\left(y,\hat y\right)}{\mathrm{Var}(y)\mathrm{Var}(\hat y)}\,,
\end{equation}
where $y$ is the vector of the predicted values. The short-term memory capacity of the network is calculated as the sum up to a certain $d_{\text{max}}$,
\begin{equation}
\mathrm{MC} = \frac{1}{d_{\text{max}}} \sum_{d=1}^{d_{\text{max}}} \mathrm{MC}_d\,,
\end{equation}
normalized so that $\mathrm{MC} \in [0,1]$. We remark that $\mathrm{MC}$ values closer to 1 indicate a higher memory capacity of the system. We fix $d_{\text{max}} = 10$ in the numerical experiments. 
This definition of memory capacity should not be confused with the one introduced in \cite{Dambre2012}. The latter also considers the capacity of non-linear memory, whereas we consider only linear short-term memory. Our choice is driven by a more practical and task-specific score for benchmarking. In addition, the total capacity of \cite{Dambre2012} is upper-bounded by the number of internal states, whereas our memory capacity depends on $d_\text{max}$.\\
The reconstruction capabilities of our echo state network are tested with the so-called NARMA task. It is a nonlinear filter with past dependence, commonly used as a benchmark for the computational power of time-dependent learning systems \cite{narma}. Precisely, the NARMA-$p$ task is formally defined as 
\begin{equation}
    \hat y_{t} = 0.4\, y_{t-1} + 0.1 \, \left(\sum_{l=0}^{p-1}y_{t-l}\right) + \, u_{t-p+1} u_t + 0.1 \,.
\end{equation}
We pick $p = 2,5 \text{ and } 8$ in the paper. The accuracy of the reconstruction is quantified with the normalized mean-square error, expressed as
\begin{equation}
\mathrm{NMSE}(y,\hat{y})= \frac{\sum_{t=T_{\text{tr} + 1}}^{t=L}|\hat y_t - y_t|^2}{\sum_{t=T_{\text{tr} + 1}}^{t=L} \hat y_t ^2}
\end{equation}
where $y$ is the vector of the predicted values after training. The input time series $u_t$ is sampled from the continuous function $f(t) = 0.1 \sin{\left(\frac{2 \pi \alpha t}{T}\right)} \cdot \sin{\left(\frac{2 \pi \beta t}{T}\right)} \cdot \sin{\left(\frac{2 \pi \delta t}{T}\right)}$ with $(\alpha, \beta, \delta, T) =(2.11, 3.73, 4.11, 200)$.
\subsection*{Ideal and realistic quantum noise models}
The effect of noise in a quantum computer can be simulated by repeatedly applying a quantum channel to the density operator that describes the state of the qubit register. 
Precisely, in numerical simulations, a quantum channel $\mathcal{E}$, representing a certain noise model, is applied to each qubit after the action of each unitary gate $U$, as shown in Equation ~\ref{uno}. 
First, we investigate the effect of three typical ideal models of quantum noise. Namely, we test the effect of amplitude damping, phase damping, and depolarizing noises. 
We recall that a quantum channel can always be expressed in terms of its Kraus decomposition,
\begin{equation}
\mathcal{E}(\rho) = \sum_{i=1}^{r}K_i\rho K_i^\dag \, .
\end{equation} 
We refer to Table \ref{noises.tab} for a detailed description of the Kraus operators of the three noise models considered in this work. The noise intensity is tuned by a parameter appearing in the Kraus decomposition, which we indicate with $\gamma$ throughout the paper. Smaller $\gamma$ values correspond to a final quantum state less affected by noise. 
After the preliminary analysis of the effect of an ideal pure noise, we test our architecture with a more realistic model of quantum noise, which faithfully emulates noise in a superconducting device. To do this, we apply the quantum channels that model the noise of IBMQ backends freely available to the user. 
The simulation of a backend noise model implemented by Qiskit~\cite{qiskit} exploits a combination of thermal relaxation and depolarizing channels.
\subsection*{Experimental settings}
In this work, we fix $N=7$ as the number of qubits. In numerical simulations, we consider a total timespan of $L=200$, with a washout period of 10 time steps ($T_{\text{wo}} = 10$), while the test interval is 60 time steps long. Each circuit is run for $S=10^5$ shots for a precise reconstruction of the average value of observables. All simulations with an ideal noise model are run over 100 values of $\gamma$. Preliminary experiments are conducted using the Qiskit simulator \cite{qiskit}. In experiments involving quantum hardware, we significantly reduce the circuit depth according to the limitations of the cloud service employing the IBM quantum service and cost management. In particular, we fix $L=80$, and for each experiment, we submit 90 jobs in batch, with 5 circuits each and 30 shots per circuit, for a total of 13500 shots. We used the IBM\_BRISBANE quantum processing unit, which mounts an Eagle processor, equipped with 127 qubits. For the simulations, we used a 36-core Intel Xeon processor (4.3 GHz) with 128 GB of RAM.
\section*{Acknowledgements} 
F.M. and E.P. are supported by PRIN-PNRR PhysiComp (nr. grant G53D23006710001). L.N. is partially funded by Eni S.p.A.
\section*{Code availability}
The code is made available by the authors upon reasonable request.
\section*{Competing interests} Authors declares no competing interests.
\bibliography{bib}
\bibliographystyle{naturemag}
\clearpage
\appendix
\onecolumngrid
\section{Supplementary Material}
 \section{Quantum reservoir computing and fading memory}\label{qrc.fm}
Denoting $\mathcal{S}(\mathcal{H})$ the set of density operators on the Hilbert space $\mathcal{H}$ that describes some quantum system, we recall that is possible to describe a quantum reservoir computer with the equations
\begin{align}\label{rc}
    \begin{cases}
        \rho_{t+1} = \mathcal T(\rho_t, u_{t+1})\\
        y_{t+1} = h(\rho_{t+1})
    \end{cases}
\end{align}
where $\mathcal I \ni u = \{\dots,u_{-1}, u_0\}$ is the input signal and $\mathcal T:\mathcal{S}(\mathcal{H}) \rightarrow \mathcal{S}(\mathcal{H})$ is a quantum channel that describes the dynamic of the reservoir. The key property of a reservoir computer is the well-known echo state property (or, convergence property) \cite{jaeger.first.long}, which ensures that for each input sequence $u$, there exists a unique reservoir signal $z$. If the echo state property holds, the reservoir computer defines a unique, causal operator $\mathcal{C}\colon \mathcal{I}\rightarrow \mathcal{O}$ defined by
\begin{equation}\label{func}
   \mathcal C(u)_t = \mathcal{C}\left(u_{|_t}\right) = h\left(\mathcal{T}(\rho_{t-1},u_t)\right) 
\end{equation}
where $u_{|_t}$ is the sequence $u$ truncated at time $t$. As a consequence, it defines a unique functional on the set of input sequences as $C(u) = \mathcal{C}(u)_0$. Note that the choice of $t=0$ as the final time step is merely irrelevant since the system is time-invariant. Fading memory is formally defined as the continuity of this functional, and thus of related filter since the relation is bijective, with respect to a proper topology.
\begin{definition}[Fading memory]\label{fm.def}
        Let $\omega\colon \mathbb{Z}^- \rightarrow(0,1]$ be an increasing function with zero limit at infinity and $\omega(0)=1$ and define the distance
        \begin{align}
        d_\omega(u,v) =   \sup_{t\in \mathbb{Z}^-}\norm{u_t-v_t}\omega_{t}, \quad \forall u,v \in \mathcal{I} \,.
        \end{align}
    A reservoir computer has fading memory if the associated functional $C: \mathcal{I}\rightarrow\mathbb{R}$ is continuous in the topology of the metric space $\left(\mathcal{I}, d_\omega\right)$.
    \end{definition}
    Fading memory of a quantum reservoir ultimately depends on the properties of the quantum channel $\mathcal T$ that defines the quantum reservoir computer. In particular, there is the following important theorem, that ensures fading memory.
    \begin{theorem}[Theorem 3 in \cite{grig2018}]
        If the quantum channel $\mathcal T$ is strictly contractive in the trace norm $\norm{A}_1 = \mathrm{tr}\left(\sqrt{A^\dag A}\right)$, namely
        \begin{equation}
            \norm{\mathcal T(\rho)-\mathcal T(\sigma)}_1 \le r \norm{\rho-\sigma}_1
        \end{equation}
        for some $r\in (0,1)$, then the associated reservoir computer has the echo state property and fading memory. 
    \end{theorem}
We remark that any quantum channel is non-expansive in trace norm, namely 
\begin{equation}
\norm{\mathcal T(\rho)}_1 \le \norm{\rho}_1\,.
\end{equation}
\section{A geometrical interpretation of non-unital channels}\label{unital.sub}
Quantum channels that preserve identity, namely $\mathcal{T}(\mathbb{I}) = \mathbb{I}$, are said to be unital. Unital channels are not suitable for reservoir computing, as shown by the following theorem. Here $\mathcal{B}(\mathcal{H})$ is the set of all bounded operators and $d$ is the dimension of the Hilbert space.
\begin{theorem}[Theorem 5 in \cite{qrc}]
    Assume there exists $\epsilon>0$ and an operator norm $\norm{\cdot}$ such that the channel $\mathcal{T}(\cdot,u)\colon \mathcal{B}(\mathcal{H})\rightarrow \mathcal{B}(\mathcal{H})$ satisfies $\norm{\mathcal T(\cdot,u)}\le 1 - \epsilon$ for any $u\in\mathcal I$. Then, the correspondent functional $C$ defined the reservoir computer in Eq.\ref{rc} is constant with $C(u) = h(\mathbb{I}/d)$ if and only if $\mathcal T$ is unital for any input, namely $\mathcal{T}(\mathbb{I}, u) = \mathbb{I} \quad \forall u \in \mathcal{I}$.
\end{theorem}
The action of unital channels has a geometrical interpretation in the Bloch sphere representation of qubit state, which allows distinguishing between unitality and non-unitality of a given channel, even without knowing its analytical expression. The set of pure states in the image of a quantum channel $\mathcal{E}$  is called its pure output $\mathrm{PO}\left(\mathcal{E}\right)$, namely
\begin{equation}
\mathrm{PO}\left(\mathcal{E}\right) = \left\{\mathcal{E}(\rho), \forall \rho \in \mathcal{S}(\mathcal{H})\right\} \cap \mathcal{P}\,.
\end{equation}
In the Bloch sphere, the pure output is the intersection between the spherical surface of the pure states and the ellipsoid representing the image of the quantum channel. Unital channels fulfill a fundamental symmetry property.
\begin{theorem}[Lemma 4.3 in \cite{usqqc}]\label{simmetry}
    The pure output of an unital channel $\mathcal{T}$ is centrally symmetric.
\end{theorem}

\begin{table}[t]
   \centering
  \begin{tabular}{>{\raggedright\arraybackslash}m{2.4cm} 
                  >{\raggedright\arraybackslash}m{3.1cm} 
                  >{\raggedright\arraybackslash}m{9.5cm}}
  \hline
  \textbf{Noise model} & \textbf{Description} & \textbf{Kraus decomposition}\\
  \hline
  Amplitude damping &  Loss of energy to the environment & $K_0 = \begin{pmatrix}
  1 & 0 \\
  0 & \sqrt{1-\gamma}
\end{pmatrix},\,\, K_1 = \begin{pmatrix}
  0 & \sqrt{\gamma}\\
  0 & 0
\end{pmatrix}$
\\ 
\addlinespace
  Phase damping & Loss of quantum information & 
  $K_0=\sqrt{1-\gamma}\,\mathbb{I}, \,\, K_1 = \begin{pmatrix}
  \sqrt{\gamma} & 0\\
  0 & 0
\end{pmatrix},$  $ K_2 = \begin{pmatrix}
  0 & 0\\
  0 & \sqrt{\gamma}
\end{pmatrix}$
 \\
 \addlinespace
  Depolarizing & Decay to maximally mixed state & $K_0 =\sqrt{1-\gamma}\,\mathbb{I}, \, \, K_1 = \sqrt{\gamma}\,X$ 
  $K_2 = \sqrt{\gamma}\,Y,\,\,K_3 = \sqrt{\gamma}\,Z$\\ 
  \hline
\end{tabular}
\caption{The Kraus decomposition of the quantum channel describing respectively amplitude damping, phase damping, and depolarizing noise.}\label{noises.tab}
\end{table}

Hence, by looking at the pure output of a quantum channel we are able to determine whether it is unital or not~\cite{usqqc}. In particular, the channel is unital if and only if the center of the image ellipsoid coincides with the center of the Bloch sphere. This fact has a trivial interpretation since the center of the Bloch sphere corresponds to the maximally mixed state $\rho=\mathbb{I}/2$. Precisely, we recall the following theorem.
\begin{theorem}[Theorem 4.9 in \cite{usqqc}]\label{teo.label}
    Let $\mathcal{E}$ be a qubit channel and $\mathcal P$ the set of pure states. One of the following holds:
    \begin{enumerate}
       \item $\mathrm{PO}(\mathcal{E})=\emptyset$, the channel has no pure output;
       \item $\mathrm{PO}(\mathcal{E})=\{\xi\}$, the channel has a unique pure output state $\xi$;
       \item $\mathrm{PO}(\mathcal{E})=\{\xi,\zeta\}$, the channel has exactly two pure outputs $\xi,\zeta$;
       \item $\mathrm{PO}(\mathcal{E})=\mathcal{P}$, all pure states are in the image of $\mathcal{E}$.
   \end{enumerate}
\end{theorem}
In particular, combining Theorem~\ref{simmetry} and Theorem~\ref{teo.label}, we can deduce that a channel with exactly two pure outputs is unital if and only if they are antipodal on the Bloch sphere (thus, are orthogonal in the Hilbert space). Fig. 5 in the Main text shows the action of amplitude damping, depolarizing, phase damping, and fake-OSAKA channel. As expected, amplitude damping is the sole non-unital ideal noise channel. Moreover, we can observe that the fake-OSAKA channel is non-unital.

\section{Contractivity of noise channels}\label{ad.app}
In this Appendix, we prove that a noise channel $\mathcal D:\mathcal{S}(\mathcal{H})\rightarrow \mathcal{S}(\mathcal{H})$ ensures fading memory by proving that it is strictly convergent in the trace norm. Namely, we show that there exists $r<1$ such that $\norm{D(\rho)- D(\sigma)}_1 \le r \norm{\rho-\sigma}_1$. For clarity, we write the proof for the single-qubit Hilbert space $\mathcal{H}=\mathbb{C}^2$. The extension to the multidimensional case is trivial. Now we propose some common noise channels and evaluate whether they guarantee fading memory or not. In order to do so, it is helpful to Kraus decompose each channel \cite{Nielsen_Chuang_2010}, hence write each channel in the form
\begin{align}
    \mathcal D(\rho) = \sum_i K_i\rho K^\dag_i.
\end{align}
For the following calculations, it is useful to keep in mind that for a generic Kraus operator $K$ we have
\begin{equation}\label{kraus_tracenorm}
    \begin{split}
        \norm{K \rho K^\dag}_1 = K_{ij}\rho_{jl}K^\dag_{li} = k_i\delta_{ij}\rho_{jl}k_l\delta_{li} = \\= k_i\rho_{ii}k_i = k_i^2\rho_{ii}
    \end{split}
\end{equation}
where summation is implied over repeated indices. Remarkably the set of strictly contractive quantum channels is dense in the set of all quantum channels \cite{strict-contr}. That is, in the assumption of finite resolution in the observations, we can assume that any channel modeling noise is strictly contractive, thus ensuring fading memory to the system.
\paragraph*{\textbf{Bit flip, bit-phase flip and phase flip.}}
The Kraus decomposition of the bit flip map is
\begin{equation}
    \mathcal{D}_\text{bf}(\rho) = \sqrt{1-\gamma}\rho + \sqrt{\gamma} X\rho X.
\end{equation}
Notice that the bit-phase flip and the phase flip maps have the same form of the bit flip, but instead of the $X$ Pauli operator, they involve the $Y$ and $Z$ operators, respectively. Hence we have
\begin{align}
    \mathcal{D}_\text{bpf}(\rho) &= \sqrt{1-\gamma}\rho + \sqrt{\gamma} Y\rho Y.\\
    \mathcal{D}_\text{pf}(\rho) &= \sqrt{1-\gamma}\rho + \sqrt{\gamma} Z\rho Z.
\end{align}
Then, for the sake of conciseness, we only show the case of the bit flip, the other cases being a trivial extension. Since Pauli operators are traceless operators, we can identify
\begin{equation}
    K_0 = \sqrt{1-\gamma}\,\mathbb{I} \qq{and} K_1 = \sqrt{\gamma}\,X
\end{equation}
By exploiting Equation~\ref{kraus_tracenorm},  we have trivially
\begin{equation}
    \norm{K_0 (\rho - \sigma) K_0^\dag}_1  \le 2\left(1 - \gamma\right)|\rho_{00} - \sigma_{00}|
\end{equation}
and
\begin{equation}
    \norm{K_1 (\rho - \sigma) K_1^\dag}_1 = 0.
\end{equation}
Thus, such channels are contractive in trace norm.

\paragraph*{\textbf{Amplitude damping.}}
The quantum operation that describes energy dissipation is the amplitude damping quantum channel.
Recalling the Kraus decomposition of the amplitude damping channel, we have
\begin{align}
    \mathcal D_\text{ad}(\rho) = K_0 \rho K_0^\dag + K_1 \rho K_1^\dag
\end{align}
with $K_0 = \begin{pmatrix}
    1 & 0  \\
    0 & \sqrt{1-\gamma}
\end{pmatrix}$ and $K_1 = \begin{pmatrix}
    0 & \sqrt{\gamma} \\
    0 & 0
\end{pmatrix}$.
We consider the two terms separately, exploiting Equation~\ref{kraus_tracenorm}.
Thus, since $\mathrm{tr}(\rho) = \mathrm{tr}(\sigma) = 1$, for the first Kraus operator we conclude that
\begin{align}
     \norm{K_0 (\rho - \sigma) K_0^\dag}_1  = |\rho_{00} - \sigma_{00}|\left(1 + \sqrt{1-\gamma}\right)\le \\ \le r \norm{\rho-\sigma} = 2|\rho_{00} - \sigma_{00}|
\end{align} 
with $r=1-\frac{\gamma}{2} < 1$. Similarly, for the other term, we get
\begin{align}
     \norm{K_1 (\rho - \sigma) K_1^\dag}_1 \le r \norm{\rho-\sigma}_1
\end{align} 
for $r = \gamma/2$. We conclude that $D$ is strictly contractive, since
\begin{align}
\norm{\mathcal D_\text{ad}(\rho)-\mathcal{D}_\text{ad}(\sigma)}_1 
\le \norm{K_0(\rho-\sigma)K_0^\dag}_1 + \\+\norm{K_1(\rho-\sigma)K_1^\dag}_1 \le r \norm{\rho - \sigma}_1
\end{align}
with $r = \max\left(1-\frac{\gamma}{2},\, \gamma/2\right) < 1$.

\paragraph*{\textbf{Depolarizing channel.}}
Consider a qubit that has probability $p$ of being depolarized.
The action of the depolarizing channel is represented by
\begin{equation}
\mathcal{D}_\text{de}(\rho)= \frac{p}{2}\mathbb{I} + (1-p)\rho.
\end{equation}
In the operator-sum representation, we can write
\begin{equation}
    \left(1-\frac{3p}{4}\right)\rho +\frac{p}{4}\left(X\rho X + Y\rho Y + Z\rho Z\right)
\end{equation}
or, equivalently, it is more commonly written as
\begin{equation}
    \left(1-p\right)\rho +\frac{p}{3}\left(X\rho X + Y\rho Y + Z\rho Z\right)
\end{equation}
which means that the state $\rho$ is left alone with probability $1-p$, and each Pauli operator acts with probability $p/3$.

\paragraph*{\textbf{Phase damping.}} 
The loss of information, without loss of energy, is described by the phase damping.
We observe that the phase damping channel can be written as
\begin{equation}
\mathcal{D}_{pd}(\rho) = \widetilde{K}_0 \rho \widetilde{K}_0 + \widetilde{K}_1 \rho \widetilde{K}_1
\end{equation}
with  
\begin{equation}
    \begin{split}
        \widetilde{K}_0 &= \mqty(1 & 0\\ 0& \sqrt{1-\lambda})\\
        \widetilde{K}_1 &= \mqty(0 & 0\\ 0 & \sqrt{\lambda})\\
    \end{split}
\end{equation}
where the parameter $\lambda$ can be interpreted as the probability of having a photon scattered (without loss of energy) as it travels through a waveguide.\\
An important thing to notice is that by a unitary recombination of $\widetilde{K}_0$ and $\widetilde{K}_1$, we can write other two equivalent Kraus operators for the phase damping channel: define $\gamma = \left(1-\sqrt{1-\lambda}\right)/2$, then
\begin{equation}
    \begin{split}
        K_0 &= \sqrt{1-\gamma}\,\mathbb{I}\\
        K_1 &= \sqrt{\gamma}\,X\\
    \end{split}.
\end{equation}
Thus, phase damping describes the same operation as the phase flip. Hence it inherits the same contraction and unitality properties.

\paragraph*{\textbf{Thermal relaxation.}}
Thermal relaxation is equivalent to phase-amplitude damping. The problem is that the Qiskit implementation of the thermal relaxation channel differs from the standard one, hence a short comment is needed.
Thermal relaxation, in the Qiskit version, is modeled by the following Kraus decomposition in Qiskit~\cite{qiskit}:
\begin{equation}
    \begin{split}
        K_0 &= \sqrt{1-p_z-p_{\rm r0}-p_{\rm r1}}\,\mathbb{I}\\
        K_1 &= \sqrt{p_z}\, Z\\
        K_2 &= \sqrt{p_{\rm r0}}\, \mqty(1 &0\\ 0 &0 )\\
        K_3 &= \sqrt{p_{\rm r1}}\, \mqty(0 &0\\ 0 &1 )
    \end{split}
\end{equation}
where $p_z$ is the probability of phase flip, $p_{\rm r0}$ is the reset probability to state $\ketbra{0}$ and $p_{\rm r1}$ is the reset probability to state $\ketbra{1}$.
\section{Universal approximation property of our gate-based echo state network}
A family of computational architectures is universal for a certain set of functionals if, for any of these mappings, there exists an instance within the family that approximates it with arbitrary precision \cite{hornik.univ.fnn}. In particular, the universality of a family of reservoir computers for fading memory mappings relies on some general sufficient conditions \cite{maass.nat} -- generalized by some of the Authors to comprehensively include quantum reservoir computing regardless of any specific implementation \cite{monz_uni}. Here, we prove that the gate-based quantum echo state network is universal under the action of nonunital and strictly contractive noise models, thus providing theoretical support to our findings. We begin by recalling the sufficient conditions that ensure the universality, summarized in the following theorem \cite{monz_uni}.
\begin{theorem}
    Let $\mathcal R$ be a family of real echo reservoir computers and let $ \mathcal{I} =\left\{u\colon\mathbb{Z}\rightarrow[0,1]\right\}$ be the set of input sequences. Assume that $\left(\mathcal{I}, d_\omega\right)$ is a compact metric space. If the set of functionals $ \mathcal{H}_{\mathcal{R}}$ associated with $\mathcal{R}:$
    \begin{itemize}
    \item has fading memory
        \item separates points in  $\mathcal{I}$, namely for any pair  $u,v\in \mathcal{I}$ with $u\not=v$ there exists a functional $H\in  \mathcal{H}_{\mathcal{R}}$ such that $H u \not = H v$;
        \item is polynomial algebra, namely for any $R_1,R_2 \in  \mathcal{H}_{\mathcal{R}}$ there exist $R^\lambda_{+}, R_{\cross} \in  \mathcal{H}_{\mathcal{R}}$ such that $  R_1(u) + \lambda R_2(u) = R_+^\lambda(u)\,and \quad R_1(u) \cdot R_2(u) = R_{\cross}(u)$ for any $u\in \mathcal{I}$;
    \end{itemize}
    then $\mathcal R$ is universal for the set of mappings with fading memory. 
    \end{theorem}
\subsection*{Compactness of the input space}
The mathematical proof of universality relies on the Stone-Weierestrass theorem, which applies to compact metric space. Thus, to ensure universality, we have to check that the input space, equipped with the fading metric introduced above is a compact metric space. This is a general result for time series inputs and it follows from the following theorem, proved also in \cite{monz_uni}. 
\begin{lemma}
   The space of bounded time-dependent sequences $\mathcal{I}$  equipped with the distance $d_\omega$ defined in Def. \ref{fm.def} is a compact metric space. 
\end{lemma}
\subsection*{Echo state property and fading memory of the associated functional}
A reservoir computer fulfills echo state property if the internal dynamics of the network does not depend explicitly on its specific initialization. This implies that the reservoir computer defines a unique, well-posed mapping between $\mathcal{I}$ and $\mathcal{O}$. As discussed above, both the echo state property and the fading memory follow if the dynamics is strictly contractive in the trace norm. Thus, these properties follow from the strict contractivity of the quantum channel that models the evolution of the echo state network, namely,
$$
 \widetilde{\rho}_{t+1} = \mathcal{E}\left(U(u_{t+1}) \rho_t U^{\dag}(u_{t+1}) \right)
$$
 since the measurement process does not erase the strict contractivity of the overall evolution. 
 \subsection*{Separability of the inputs}
 A universal class of reservoir computers has the ability to discriminate any pair of different inputs. Namely, for any given $u,v \in \mathcal{I}$, the dynamics of a reservoir in the class should be such that associated mapping has different outputs when evaluated on $u$ and $v$, respectively. We proceed by constructing it,  exploiting a single qubit reservoir and the action of suitable quantum channels. Denoting with $x\in [0,1]$ the current value of the input, the dynamics of reservoir register is written as
\begin{equation}\label{sing.qub}
    \rho_{t+1} =  \Phi_x(\rho_t) = \mathcal{E} \left(U(x) \rho_t U^\dag (x) \right) \,.
\end{equation}
Then, we can prove the following Theorem that ensures the separability of the input under suitable conditions on the dynamics.
\begin{theorem}
    If the dynamics described by Eq.~\eqref{sing.qub} has a unique fixed point $\rho_*^x$, that depends univocally on $x$, then we have separability.
\end{theorem}
 In fact, without loss of generality, let assume $u,v \in \mathcal I$ such that $u_t = v_t \,\,\forall t\not=0$ and $u_0 = v_0$, and let denote with $\rho_t^y,\, y=u,v$ the state of the circuit after the injection of $u$ and $v$, respectively. Since the input sequences are equal for $t<0$, we assume moreover that $\rho^u_0 = \rho^v_0$. Preparing the circuit in the fixed point $\rho_*^{u_0}$, immediately we realize that $\rho_1^u \not = \rho_1^v$ since the fixed point $\rho_1^u  = \rho_*^{u_0} $ depends univocally on $u_0$.  Then the dynamics continues separated, as indeed the linear operator $e^{\mathcal{L}_t(y)}$, that encodes the evolution in Eq. \eqref{sing.qub}, is a linear full rank operator. \\
It remains to show that there exists a quantum channel $\mathcal{E}$ such that Eq.~\eqref{sing.qub} has a unique fixed point. This is the case for the amplitude damping channel.
\begin{lemma}
    If the quantum channel in Eq.~\eqref{sing.qub} is the amplitude damping $\mathcal{E}_{\text{ad}}$, then the equation has a unique fixed point that depends univocally on the value of the input.
\end{lemma}
\begin{proof}
    It suffices to recall that the unique fixed point of the amplitude damping channel is the $|0\rangle$ state, namely $\mathcal{E}_{\text{ad}}(\rho) = \rho$ if and only if $\rho =|0\rangle \langle 0| $. Then, the fixed point of the composed map $\rho_{t+1} = \mathcal{E}_{\text{ad}}\left(U(x) \rho_t U^\dag (x) \right) $ is simply $\rho_*^x = U^\dag(x) |0\rangle \langle 0| U(x) $, which depends univocally on $x$.
\end{proof}
We remark that, on the other hand, if the quantum channel $\mathcal E$ is unital, namely if it preserves the identity density operator $\mathcal{E}\left(\mathbb{I}\right) = \mathbb{I}$, the dynamics has a fixed point that does not depend on the input values. Indeed, by definition of the unitality of a quantum channel, we have that 
$$
\Phi_x(\rho_*) = \rho_* \iff \rho_* = U(x)\mathbb{I}U^\dag(x) = U(x)U^\dag(x)\mathbb{I} = \mathbb{I}\,.
$$
regardless of the value of the input $x$.
 \subsection*{Polynomial algebra structure by spatial multiplexing}
  It remains to prove that the associated mappings form a polynomial algebra, to ensure the universality of our family of quantum echo state networks. This can be achieved by exploiting spatial multiplexing, as widely known in the context of quantum reservoir computing \cite{chen2, nokkala2021gaussian, zamb_dissi}.  It consists of parallelly preparing and running the two networks and taking the sum, or the product, of the readouts as the overall readout function. Then, for any two echo state networks whose associated mapping are respectively $\mathcal{C}_1$ and $\mathcal{C}_2$, it is possible to construct the systems whose associated mapping is the sum $\mathcal{C}_1 + \mathcal{C}_2$ and the product $\mathcal{C}_1 \cdot \mathcal{C}_2$. Adding this multiplexed architecture to the general class, we recover the structure of the polynomial algebra of the mappings associated. Thus, combining the previous Subsections, we have proved that, under suitable conditions on the quantum channel - for example, fulfilled by the amplitude damping channel -  our class of gate-based echo state networks described is universal for the class of filters with fading memory.
\end{document}